\documentclass[conference]{IEEEtran}
\IEEEoverridecommandlockouts
\usepackage[T1]{fontenc}

\newcommand{\Tsupersqueezeequ}{\medmuskip=0.1mu \thinmuskip=0mu \thickmuskip=0.1mu \nulldelimiterspace=-1pt \scriptspace=0pt}

\ifCLASSINFOpdf
\else
\fi
\usepackage{blkarray}
\usepackage{multirow}
\newcommand{\rdots}{\hspace{.2ex}\raisebox{1ex}{\rotatebox{-12}{$\ddots$}}}

\usepackage{cite}
\usepackage{amsmath,amssymb,amsthm,amsfonts}
\usepackage{mathtools}
\usepackage{graphicx,caption,subcaption}
\usepackage{textcomp}
\usepackage{xcolor}
\usepackage{bm}

\usepackage[nameinlink]{cleveref}
\def\BibTeX{{\rm B\kern-.05em{\sc i\kern-.025em b}\kern-.08em
		T\kern-.1667em\lower.7ex\hbox{E}\kern-.125emX}}
\newtheorem{thm}{Theorem}
\newtheorem{lem}[thm]{Lemma}

\DeclarePairedDelimiterX{\norm}[1]{\lVert}{\rVert}{#1}
\newtheorem{defn}{Definition}

\usepackage{algorithm,
	algpseudocode}
\usepackage[cmintegrals]{newtxmath}
	\usepackage[normalem]{ulem}
	\usepackage{placeins}
	\usepackage[displaymath, mathlines]{lineno}
\begin{document}
	
	\title{Local Approximation of Secrecy Capacity}
	
 \author{%
	   \IEEEauthorblockN{Emmanouil M.~Athanasakos\IEEEauthorrefmark{1},		                     Nicholas~Kalouptsidis\IEEEauthorrefmark{2} and Hariprasad Manjunath\IEEEauthorrefmark{1}	                  
		                   }
	   \IEEEauthorblockA{\IEEEauthorrefmark{1}%
		                   INRIA, Centre Inria d’Université Côte d’Azur, Sophia Antipolis, France.}
	   \IEEEauthorblockA{\IEEEauthorrefmark{2}%
                            National and Kapodistrian University of Athens, Department of Informatics and Telecommunications, Athens, Greece.}
		                     
	 }
	
	\maketitle
	
	\begin{abstract}
		This paper uses Euclidean Information Theory (EIT) to analyze the wiretap channel. We investigate a scenario of efficiently transmitting a small amount of information subject to compression rate and secrecy constraints. We transform the information-theoretic problem into a linear algebra problem and obtain the perturbed probability distributions such that secrecy is achievable. Local approximations are being used in order to obtain an estimate of the secrecy capacity by solving a generalized eigenvalue problem.
	\end{abstract}
	
	
	\section{Introduction}\label{sec_I}
	Security and privacy of information flows has become a critical requirement in system design. Currently, security is handled as an add on by the higher network layer protocols. However, with the thriving of the internet and the fast growth of mobile communications, reliance on this layered approach is inefficient. Incorporating security constraints as an inherent module by design, offers a framework in which security of information can be determined with quantitative metrics stemming from information theory.\par 
	Quantification of security with information-theoretic metrics is a well known challenge. Moreover, the study of secrecy rates that are achievable by practical coding schemes constitute a core subject of physical layer security~\cite{kolok_1} - \cite{ath_polar}. Although the maximal rate of secure transmission for the point-to-point case has been characterized, we show in this work that by viewing this problem through local approximations leads to easier to compute solutions. Finally we demonstrate that our solutions apply to the equivalent problem of privacy-utility trade off that includes as special cases the information bottleneck~\cite{bottleneck} and the privacy funneling problem \cite{funnel}.
	\subsection{Contributions}
	Specifically, we explore an alternative path which is built upon the \textit{Euclidean local approximations} proposed in~\cite{Borade_1}, ~\cite{Zheng_1} and~\cite{Zheng_2}. While our model consists of the basic building block for modeling secure transmission over noisy channels, we deviate from the usual treatment for proving achievability and instead solve an optimization problem. We place our setting in an environment where users transmit small amount of information and we impose constraints on the transmission rate and the information leakage. The security constraint corresponds to the strong secrecy requirement~\cite{Maurer} and we show that by assuming distributions around a specific neighborhood we can approximate the secrecy capacity by transforming the mutual information terms into squared weighted Euclidean norms.
	\subsection{Related Work}
	Euclidean local approximations were introduced \textcolor{black}{in \cite{Borade_1} and later extended} by Huang-Zheng in \cite{Zheng_1} and \cite{Zheng_2} as a way to acquire solutions for network information-theoretic problems through linear algebraic methods and to obtain single-letter solutions. This is possible due to a linearization technique where the distributions of interest are close to each other in terms of Kullback–Leibler~(KL) divergence, thus imposing a linear assumption on the distribution space instead of treating it as a manifold. Recently, EIT has gained popularity in the area of statistical inference and machine learning \cite{Huang_1}, \cite{Huang_2} and \cite{Makur_20}. \par 
	Euclidean local approximation-based privacy preserving mechanisms have been developed recently for controlling the data disclosure. In particular, the authors of \cite{Tan_1} inspired by the Chernoff-Stein lemma, formulated the privacy-utility (PUT) problem for hypothesis testing using relative entropy or R\'{e}nyi divergence to measure the utility function and mutual information as a privacy metric. To solve this problem, they develop local approximations and perturb the privacy mechanism around a perfectly private operation point for each pair of hypotheses, and transform it into a semi-definite program. In \cite{Tan_2}, the authors utilize Euclidean approximations to find the optimal error exponent, in a distributed hypothesis testing scenario. A similar line of work \cite{Zhou}, addresses the problem of characterizing the PUT trade-off between the information leakage and the misclassification error, when there is a noisy channel between the detector and the transmitter. As in \cite{Tan_1}, they consider the EIT approach for the high privacy regime and showed that for very noisy channels the approximations are very tight. In \cite{Zamani_eit}, the problem of designing a privacy mechanism is transformed into a linear program via the local geometry approach and the privacy is measured with the $l_1$ norm. A closely related work is~\cite{Gunduz}, where the authors consider the PUT problem under a rate constraint. Specifically, the case of perfect privacy is studied and solved with a linear program. We discuss in detail the differences with this work in Sec.\ref{sec_II}. 
	%
	\section{Preliminaries And Notations}\label{sec_III}
	In this section, we present the required background and set up the notation needed for the theoretical analysis of the following sections.
 
	%
 
	Assuming that the distributions under consideration are very close to each other allows us to approximate~(locally) the space of distributions by the tangent space through a linearization technique. This allows the description of information-theoretic problems using linear algebraic arguments whereas the KL divergence can be viewed as valid metric on the space of probability distributions.
	Formally, let $P_X$ and $Q_X$ be two probability distributions over a finite alphabet $\mathcal{X}$~($|\mathcal{X}|<\infty$) and assume they are close, in the sense that $Q_X^{(\epsilon)}$ is a perturbation of $P_X$, that is $Q_X^{(\epsilon)} = P_X+\epsilon J$.  $\epsilon\in(0,1)$ so that $\forall x \in \mathcal{X}, 0\leq Q_X^{(\epsilon)}\leq 1$.
		Moreover, the additive perturbation vector $J_u$ is such that
		\begin{align}
			\sum_{x\in \mathcal{X}}J_u(x) &= 0 ,\quad \forall u \in \mathcal{U} \label{valid_prob_1} \\
			\sum_{u\in \mathcal{U}}P_U(u) \cdot J_u(x) &= 0, \quad \forall x\in \mathcal{X} \label{valid_prob_2}
		\end{align} 
		where~\eqref{valid_prob_1} guarantees that $Q_X^{(\epsilon)}$ is a genuine probability vector; and~\eqref{valid_prob_2} secures that the marginal pmf of $X$ is preserved, i.e., $\sum_{u\in \mathcal{U}}P_U(u) P_{X|U=u} = P_X$. Equality in~\eqref{valid_prob_2} in vector form satisfies 
		\begin{equation}
			\sum_{u\in \mathcal{U}}P_U(u)J_u=0. \label{valid_prob_3}
	\end{equation} 
	\textcolor{black}{Let the relative interior of the probability simplex be denoted as $\text{relint}\{\mathcal{P}(\mathcal{X})\}$, and assume that $P_X \in \text{relint}\{\mathcal{P}(\mathcal{X})\}$, i.e., all entries $P_X(x)$ are positive and sum to one. Let the KL divergence between two distributions $P_X$ and $Q_X$ be
		\begin{equation}
            \label{KL_divergence_def}
			D(P_X||Q_X)=\sum_{x\in {\cal X}}P_X(x)\log\frac{P_X(x)}{Q_X(x)},
		\end{equation}
		where standard conventions\footnote{$0\log\frac{0}{0}=0$ and if there is $x$ such that $Q_X(x)=0$, then the KL divergence is plus infinity.} hold,~\cite{Polyanskiy_notes}. Given a perturbation vector $J$, $\epsilon$ can be restricted in an interval $(0, r)$ where $Q_X^{(\epsilon)}(x)>0$ for all $x\in \mathcal{X}$. Indeed take $r=\min_x P_X(x)/|J(x)|$. If $J(x)\geq 0$ clearly $Q_X^{(\epsilon)}(x)>0$. If $J(x)<0$, then $\epsilon<r<P_X(x)/(-J(x))$ and $Q_X^{(\epsilon)}(x)>0$. As $\epsilon$ ranges over the interval $(0,r)$ it holds }
        \label{proof_KL}
        \begin{align}
                 \nonumber
                D(P_X||Q_X^{(\epsilon)}) &= \sum_{x\in\mathcal{X}}P_X(x)\log\frac{P_X(x)}{Q_X^{(\epsilon)}(x)}\\ 
                \label{proof_KL_taylor1} 
                &= - \sum_{x\in\mathcal{X}}P_X(x)\log \big(1+ \epsilon \frac{J(x)}{P_X(x)}\big)\\     
                 \label{proof_KL_taylor2} 
                & = \frac{1}{2}\epsilon^2 \sum_{x\in\mathcal{X}} \frac{J(x)^2}{P_X(x)} + o(\epsilon^2),  
        \end{align}
	where equality in~\eqref{proof_KL_taylor1} is derived using the second order Taylor approximation of the natural logarithm, i.e., $\log(1+x) = x - \frac{x^2}{2} + o(x^2)$, and $o(\epsilon)$ denotes the Bachmann-Landau notation\footnote{Describes the limiting behavior of a function $f:\mathbb{R}^+ \to \mathbb{R}$ such that $\frac{f(\epsilon)}{\epsilon}\to 0$ as $\epsilon \to 0$}. 
	\textcolor{black}{The equality in~\eqref{proof_KL_taylor2} implies that the KL divergence is approximated by the weighted squared norm of the perturbation vector $J$ where the weights are given by $1/P_X(x)$. Moreover this indicates that $D(P_X||Q_X^{(\epsilon)})$ is locally symmetric and thus can be approximately viewed as a valid norm in a neighborhood of $P_X$, i.e.,
		\begin{equation}
            \label{KL_additive_per}
			D(P_X||Q_X^{(\epsilon)})= \frac{1}{2}\epsilon^2 \parallel J \parallel^2_P + o(\epsilon^2)
			=D(Q_X^{(\epsilon)}||P_X).
		\end{equation}
		It will be convenient for later purposes to change coordinates and pass to the spherical perturbation
		\begin{align}
            \label{L_spherical_pert}
			L= [\sqrt{P}]^{-1} J, 
		\end{align}
    where $[\sqrt{P}]^{-1}$ in~\eqref{L_spherical_pert} denotes the inverse of the diagonal matrix with entries $\{\sqrt{P(x)}, x\in \mathcal{X}\}$. Note that from~\eqref{L_spherical_pert}, constraint in~\eqref{valid_prob_1} takes the form $\sum_{x\in\mathcal{X}}\sqrt{P(x)}L(x) = 0$, or $L^{\top} \sqrt{P_X}=0$, while the KL divergence becomes the squared norm of $L$:
		\begin{align}
			D(P_X||Q_X^{(\epsilon)}) &= \frac{1}{2}\epsilon^2 \parallel L \parallel^2 + o(\epsilon^2). \label{KL_spherical_per}
        \end{align}}
	The following definitions from~\cite{Huang_1} and~\cite{Huang_2} aids in the characterization of the main optimization problem considered in this work.
	\begin{defn}\label{Def_Eregion}
		(\textit{$\epsilon$-region}) Given a distribution $P_X \in \text{relint}\{\mathcal{P}(\mathcal{X})\}$, the set of all distributions in the $\epsilon$-region of $P_X$ is such that
		\begin{align}
          \label{set_arount_e_region}
			\mathcal{N}_\epsilon^\mathcal{X} \overset{\Delta}{=} \{P'_X \in \mathcal{P}(\mathcal{X}) : \chi^2(P'_X,P_X) \leq \epsilon^2\},
		\end{align}
	where $\chi^2(P_X,P'_X) = \sum_{x\in\mathcal{X}} \frac{(P'_X(x)- P_X(x))^2}{P_X(x)}$.
	\end{defn}
	\begin{defn}\label{Def_Divergence_matrix}
		(\textit{Divergence transfer matrix}) Let $X$ and $Y$ be two random variables over finite alphabets $\mathcal{X}$ and $\mathcal{Y}$, respectively. Suppose that $X$ is the input in a discrete memoryless channel with output $P_Y \in \text{relint}\{\mathcal{P}(\mathcal{Y})\}$, and conditional probability distributions $P_{Y|X}$. The divergence transfer matrix (DTM) associated with $P_{Y|X}$ is defined as follows:
		\begin{align}
        \label{DTM_1}
			B_{Y|X} \overset{\Delta}{=} [\sqrt{P_Y}]^{-1} P_{Y|X} [\sqrt{P_X}] , 
		\end{align}
		where $P_{Y|X}$ denotes the $|\mathcal{Y}| \times |\mathcal{X}|$ left stochastic transition probability matrix. 
		
	\end{defn}
	\noindent The following Lemma stems from Definitions~\ref{Def_Eregion} and~\ref{Def_Divergence_matrix} and a sequence of algebraic manipulations based on equality in~\eqref{proof_KL_taylor2}. 
	\begin{lem}\label{prop_1}
     Given a distribution $P_{X|U=u} \in \mathcal{N}_\epsilon^\mathcal{X}(P_X)$, for all $u \in \mathcal{U}$ and all $\epsilon>0$, it holds that the mutual information for the encoding satisfies
     \begin{subequations}
         \begin{align}
         \label{prop_1_res}
			I(U;X)=\frac{\epsilon^2}{2}\sum_{u\in \mathcal{U}} P_U(u) \cdot \parallel L_u \parallel^2+o(\epsilon^2).
		\end{align}
        The mutual information between the input and the output of the channel $P_{Y|X}$ satisfies
        \begin{align}
        \label{prop_2_res}
			I(U;Y)= \frac{1}{2}\epsilon^2 \sum_{u\in \mathcal{U}} P_U(u)\parallel B_{Y|X}L_u \parallel^2 + o(\epsilon^2),
		\end{align}	
        and the information leakage of the channel $P_{Z|X}$ satisfies
        \begin{align}
        \label{prop_3_res}
			I(U;Z)= \frac{1}{2}\epsilon^2 \sum_{u\in \mathcal{U}} P_U(u)\parallel B_{Z|X}L_u \parallel^2 + o(\epsilon^2) 
		\end{align}
     \end{subequations}
 where $L_u = [\sqrt{P_X}]^{-1}\cdot J_u$ is the spherical perturbation vector in $\mathbb{R}^{|\mathcal{X}|}$; $B_{Y|X}$ and $B_{Z|X}$ are the divergence transition matrices defined in~\eqref{DTM_1} of the legitimate~(main) channel and of the eavesdropper's~(wiretap) channel, respectively.
	\end{lem}
 \begin{proof}
     The proof is presented in Appendix A. 
 \end{proof}
	\section{Problem Formulation} \label{sec_II}
	We consider the standard wiretap channel model in which Alice transmits a message $U\in \mathcal{U}$ over a discrete memoryless wiretap channel (DM-WTC) with input $X\in \mathcal{X}$ and outputs $Y \in \mathcal{Y}$, $Z \in \mathcal{Z}$, as shown in Fig.~\ref{nestedpc}. Secure communication is established when Alice sends a message reliably to Bob, while Eve receives no useful information regarding the transmitted message. Formally, the secrecy capacity for the one-shot version of a physically degraded wiretap channel, namely a channel for which the eavesdropper channel is worse than the legitimate link is given by:
	\begin{align}
		C_s = \max_{P_X} I(X;Y) - I(X;Z), \label{sec_cap}
	\end{align}
	where $X \to Y\to Z$ forms a Markov chain and thus channel prefixing is not needed, as choosing $X$ is optimal~\cite{Csiz_Kor}. Usually, achievability is ensured by the stochasticity of the encoding process in combination with codebook binning~\cite{Blo_Bar}. Coding binning  aims to create a confusion rate $I(X;Z)$, while guaranteeing successful decoding at rates up to $I(X;Y)$. Next we view the secrecy capacity under the prism of the following optimization problem
	\begin{subequations}
		\label{eq:optim1}
		\begin{align}
			C_s=\underset{P_X}{\text{max}}
			& \quad  I(X;Y)   \label{eq:cost1}\\
			\text{subject to} 
			& \quad I(X;Z) \leq  \Theta, \label{eq:const1_2}
		\end{align}
	\end{subequations}
	\begin{figure}[!htb]
		\centering
		\def\svgwidth{.7\linewidth}
		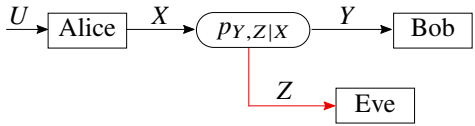
		\vspace{0.02cm}
		\caption{The wiretap channel model.}
		\label{nestedpc}
	\end{figure}
	where $\Theta$ is a small and positive threshold on the leakage rate such that the adversary's best strategy is to guess the message. Note that problem in~\eqref{sec_cap} has the same operational meaning with the one defined in~\eqref{eq:optim1} by not allowing leakage rates above $\Theta$. We now go one step further and propose a different formulation influenced by the local approximation method~\cite{Huang_1} and its extension to multi-user networks~\cite{Huang_2}. For a fixed input distribution $P_X \in \text{relint}\{\mathcal{P}(\mathcal{X})\}$, we seek to find $p_U$ and $p_{X|U}$ that are consistent with the Markov assumption and maximize the mutual information over the legitimate link subject to rate and leakage constraints:
	\begin{subequations}
		\label{eq:optim2}
		\begin{align}
			\underset{P_U , P_{X|U}}{\text{max}}
			& \quad  I(U;Y)   \label{eq:cost2}\\
			\text{subject to} 
			& \quad I(U;X) \leq  R, \label{eq:const2_1}\\
			& \quad I(U;Z) \leq  \Theta, \label{eq:const2_2}
		\end{align}
	\end{subequations}
	where $U \to X \to (Y,Z)$ forms a Markov chain.
 
    Maximizing $I(U;Y)$ can be interpreted as the maximization of the data rate of information transmitted from Alice to Bob, whereas constraint \eqref{eq:const2_1} is the compression rate of information modulated in $X$, which as noted in \cite{Zheng_1}, can be interpreted as how efficiently we can send a small amount of information through the legitimate channel instead of solving the problem of how many bits in total we can transmit. Constraint \eqref{eq:const2_2} controls the information leakage, which ideally is close to zero, that is Eve's observation $Z$ is statistically independent from the transmitted message $U$. Note that if $R\leq \Theta$, the data processing inequality $I(U;Z)\leq I(U;X)$ and therefore the leakage constraint always holds and can be removed. In practice, we typically expect $\Theta<R$ and therefore the imposition of both constraints is meaningful.		
	Problem \eqref{eq:optim2} is equivalent to the privacy-utility tradeoff between two parties, an agent and a utility provider. Indeed, in \cite[Problem (2)]{Gunduz} the agent wishes to maximize the shared information with the utility provider under a rate constraint whereas leakage does not exceed a prespecified threshold. The random variable $X$ is observed by the agent and is correlated with $Y$ and $S$. Considering the Markov chain $(Y,S) \to X \to U$ and substituting $S=Z$ and $\Omega= \Theta$,  in conjunction with the fact that the Markov property is preserved when the direction of information flow is reversed, shows that \eqref{eq:optim2} and the problem in \cite{Gunduz} are indeed equivalent.  
	\section{Local Approximation of Secrecy Capacity}
	Adopting the local approximations from Lemma~\ref{prop_1}, the maximization problem of~\eqref{eq:optim2} is transformed into the following linear algebra problem
	\begin{subequations}
			\label{eq:optim3}
			\begin{align}
				\underset{P_U,L_u}{\text{max}}
				& \quad  \sum_{u\in \mathcal{U}} P_U(u)\parallel B_{Y|X}L_u \parallel^2, \label{eq:cost3}\\
				\text{subject to} 
				& \quad \sum_{u\in \mathcal{U}} P_U(u) \cdot \parallel L_u \parallel^2 \leq \frac{2R}{\epsilon^2}, \label{eq:const3_1}\\
                & \quad  \sum_{u\in \mathcal{U}} P_U(u)\parallel B_{Z|X}L_u \parallel^2 \leq \frac{2\Theta}{\epsilon^2}, \label{eq:const3_4} \\
				& \quad \sqrt{P_X}^{\top} L_u = 0, \quad \forall u, \label{eq:const3_2}\\
				& \quad \sum_{u\in \mathcal{U}}P_U(u)[\sqrt{P_X}]L_u = \mathbf{0}, \label{eq:const3_3}\\				
                & \quad  \sum_{u\in \mathcal{U}} P_U(u)=1. \label{eq:const3_5} 
			\end{align}
		\end{subequations}
        Note that if the matrix $[\sqrt{P_X}]$ is of full rank, it holds that the constraint in~\eqref{eq:const3_3} is equivalent to $\sum_{u\in \mathcal{U}}P_U(u)L_u = 0$. Additionally, in problem~\eqref{eq:optim2} we ignored all the $ o(\epsilon^2)$ terms and explicitly impose the additional constraints~\eqref{eq:const3_2},~\eqref{eq:const3_3} and~\eqref{eq:const3_5} which guarantee the validity of the perturbed pmf.\\
	%
    Let $V$ and $\Lambda$ be two symmetric and non-negative definite matrices, such that $V=B_{Y|X}^{\top} B_{Y|X}$ and $\Lambda = B_{Z|X}^{\top} B_{Z|X}$, with $B_{Y|X}$ and $B_{Z|X}$ in~\eqref{DTM_1}. The optimization problem in~\eqref{eq:optim3} is approximately equivalent to the following: 
		\begin{subequations}
				\label{eq:optim5}
				\begin{align}
					\underset{P_U,{L_u}}{\text{min}}
					& \quad  -\sum_{u\in\cal{U}}P_U(u)L_u^{\top}VL_u \label{eq:cost5}\\
					\text{subject to} 
					& \quad\sum_{u\in\cal{U}}P_U(u)L_u^{\top}L_u \leq \frac{2R}{\epsilon^2}, \label{eq:const5_1}\\
					& \quad \sum_{u\in\cal{U}}P_U(u)L_u^{\top}\Lambda L_u \leq \frac{2\Theta}{\epsilon^2}, \label{eq:const5_4} \\
					& \quad\sqrt{P_X}^{\top}L_u=0 ,  \qquad  \forall u, \label{eq:const5_2}\\
					& \quad \sum_{u\in \mathcal{U}}P_U(u) [\sqrt{P_X}] L_u =  \mathbf{0},  \label{eq:const4_3}\\
                     & \quad  \sum_{u\in \mathcal{U}} P_U(u)=1, \qquad P_U(u)\geq 0. \label{eq:const5_5} 
				\end{align}
			\end{subequations}
		For a given distribution $P_U \in \Delta(\mathcal{U})$, the resulting optimization is a quadratically constrained quadratic program~(QCQP), denoted by $\texttt{QCQP}(\bm p_0,V,\Lambda,P_X,R,\Theta, \bm L_0)$. Additionally, given a vector ${L_u}$, minimizing over $P_U$ amounts to a linear program~(LP), denoted by $\texttt{LP}(\bm L_1,V,\Lambda,R,\Theta,\bm p_0)$. The combination of the two steps above leads to an iterative alternating procedure illustrated in Algorithm~\ref{altoptim}.
    \begin{algorithm}
	\algrenewcommand\algorithmicrequire{\textbf{Input:}}
	\algrenewcommand\algorithmicensure{\textbf{Output:}}
	\caption{\texttt{Alternating optimization for $P_U$ and $L_u$}}\label{altoptim}
	\begin{algorithmic}[1]
		\Require $V, \Lambda, P_X, R, \Theta, \epsilon$.
		\Ensure  $P_U, L_u$
       \State $\text{Initialize :\quad} \bm p_0$, $\bm L_0$ , $\text{tol} = 10^{-6}, \text{pErr} = 1$   
        \While{$ \text{pErr} < \text{tol} $} \Comment{ }
       \State $\bm L_1 = \texttt{QCQP}(\bm p_0,V,\Lambda,P_X,R,\Theta, \bm L_0)$
       \State $\bm p_1 = \texttt{LP}(\bm L_1,V,\Lambda,R,\Theta,\bm p_0)$
       \State $\text{pErr} = \norm{\bm p_1 - \bm p_0}$
       \State $\bm L_0 = \bm L_1$ , $\bm p_0 = \bm p_1$
        \EndWhile
       \State $P_U = \bm p_1 $, $[L_u] = \bm L_1$
	\end{algorithmic}
\end{algorithm}

\noindent A detailed analysis of the alternating optimization approach will be available in the extended version of this work. 
    
    \noindent \textcolor{black}{ Consider the Lagrangian associated with problem \eqref{eq:optim5}:
		\begin{align}
			&\mathcal{L}(L_u,\nu,\rho,\pmb\xi,\pmb\mu) \nonumber\\
			& = -\sum_{u\in \mathcal{U}} P_U(u) L_u^{\top} V L_u  + \nu(\sum_{u\in \mathcal{U}} P_U(u)L_u^{\top} \Lambda L_u-\Theta) \nonumber \\ 
			&+ \rho(\sum_{u\in \mathcal{U}} P_U(u)L_u^{\top} L_u - R) + \sum_{u\in \mathcal{U}} \xi(u) \sqrt{P_X}^{\top} L_u \nonumber \\
			&+ \sum_{u\in \mathcal{U}} P_U(u) \pmb\mu^{\top} L_u, \label{EqLagrangian_1}		
		\end{align}	
		where $\pmb\xi$, $\pmb\mu$ are vectors of size $|\mathcal{U}|$ and $|\mathcal{X}|$ respectively. It is assumed that a constraint qualification holds and the KKT conditions are applicable. Let,
		\begin{align}
  \label{kmatrix}
			K(\nu,\rho)= - V+\nu \Lambda + \rho I,
		\end{align}
  where $I$ denotes the identity matrix.
		Note that the dual function $\mathcal{G}(\nu,\rho,\pmb\xi,\pmb\mu) = \inf_{L_u}\mathcal{L}$ is finite and positive, i.e., $0<\mathcal{G}<\infty$, provided that $K \succeq 0$. The KKT conditions lead to the following result.
		\begin{thm}\label{theorem_1} Let $L^{\star}_u$ and $\nu^{\star},\rho^{\star},\pmb\xi^{\star},\pmb\mu^{\star}$ be the optimal solutions of the primal and the dual problem, respectively. Then
			\begin{align}
					\pmb\mu^{\star} &= - \left(\sum_{{u'}\in \mathcal{U}} \xi_{u'}\right) \sqrt{P_X}, \label{Eq_theorem_1}\\
					\pmb\xi^{\star} &= \left(\sum_{{u'}\in \mathcal{U}} \xi_{u'}\right) P^{\star}_U(u) \quad ;\text{and} \label{Eq_theorem_2}\\
					\forall u \in \mathcal{U} \quad K(\nu^{\star},&\rho^{\star})L^{\star}_u = 0. \label{Eq_theorem_3}
				\end{align}
    where the matrix $K$ is in~\eqref{kmatrix}.
		\end{thm}
		\begin{proof}
			The proof is provided in Appendix B.
		\end{proof}
	}
	%
	%

	\subsection{Example: Binary Symmetric Wiretap Channel}
	Next we provide an illustrative example. We consider a BSC(p) as the legitimate channel and a BSC(q) as the eavesdropper's channel. The channel transition probability matrices $W^{(m)}$ and $W^{(e)}$ are:
	\begin{equation*}
		W^{(m)}=
		\begin{bmatrix}
			1-p & p \\
			p & 1-p 
		\end{bmatrix} \quad \text{and} \quad 
		W^{(e)}=
		\begin{bmatrix}
			1-q & q\\
			q & 1-q
		\end{bmatrix} 
	\end{equation*}
		
	and $P_X$ is $[\frac{1}{2},\frac{1}{2}]^{\top}$. Now we can calculate $P_Y$ and the corresponding DTM $B_{Y|X}$ as
	\begin{align}
		P_Y &= W^{(m)} P_X = [\frac{1}{2},\frac{1}{2}]^{\top} \quad \text{and} \nonumber\\
		B_{Y|X} & =  [\sqrt{P_Y}^{-1}]W^{(m)} [\sqrt{P_X}] = W^{(m)},\nonumber
	\end{align}
	similarly we compute $P_Z$ and $B_{Z|X}$. The resulting DTMs are symmetric and non-negative definite, thus by the spectral theorem the eigenvalues are non-negative real numbers. We then have:
	\begin{equation*}
		V={B_{Y|X}}^{\top} B_{Y|X}=
		\begin{bmatrix}
			(1-p)^2+p^2 & 2p(1-p) \\
			2p(1-p) & (1-p)^2+p^2 
		\end{bmatrix},
	\end{equation*}
	\begin{equation*}
		\Lambda={B_{Z|X}}^{\top} B_{Z|X}=
		\begin{bmatrix}
			(1-q)^2+q^2 & 2q(1-q) \\
			2q(1-q) & (1-q)^2+q^2 
		\end{bmatrix} 		
	\end{equation*}
	\textcolor{black}{The $2 \times 2$ dimensional symmetric matrices $V, \Lambda$ share a common eigenvalue-eigenvector pair $(1,\sqrt{P_X})$ and there is a unique vector $\tau$ of length one that is orthogonal to $\sqrt{P_X}$. Then the matrix $Q = [\sqrt{P_X}, \tau ]$ necessarily diagonalizes $V$ and $\Lambda$ simultaneously, i.e.
		\begin{align}
			V = Q \Delta_{V} Q^{\top} \quad \text{and} \quad \Lambda = Q \Delta_{\Lambda} Q^{\top},
		\end{align}
		where $\Delta_{V} = \text{diag}\{1,\lambda_V\}$ and $\Delta_{\Lambda} = \text{diag}\{1,\lambda_{\Lambda}\}$. Recall that $K=Q \Delta Q^{\top}$ with $\Delta = \text{diag}\{\nu +\rho  - 1, \nu \lambda_{\Lambda}+\rho -\lambda_V \}$, with $\nu +\rho  - 1\geq 0 $, $\nu \lambda_{\Lambda}+\rho -\lambda_V \geq 0$ and at least one of the two equal to zero.}\newline
	\textcolor{black}{Suppose that $\nu \lambda_{\Lambda}+\rho -\lambda_V > 0$, then condition \eqref{Eq_theorem_3} from Theorem~\ref{theorem_1} implies
		\begin{align}
			\begin{bmatrix}
				\nu +\rho  - 1 & 0 \\
				0 & \nu \lambda_{\Lambda}+\rho -\lambda_V 
			\end{bmatrix} \times \begin{bmatrix} \sqrt{P_X}^{\top} L^{\star}_u \\ \tau^{\top} L^{\star}_u \end{bmatrix} = 0,
		\end{align}
		which implies $\tau^{\top} L^{\star}_u = 0$ or $L^{\star}_u = 0$ which is a contradiction. Hence it follows that $\nu \lambda_{\Lambda}+\rho = \lambda_V$,
		and constraint \eqref{Eq_theorem_3} holds for any $L_u$. Checking the remaining constraints leads to $L_u = s(u) \tau$, which follows from the orthogonality constraint $\sqrt{P_X}^{\top} L_u =0$ and the observation that $L_u$ must be colinear with vector $\tau$. Therefore, the other constraints become
		\begin{align*}
			\sum_u P_U(u) s(u) &= 0, \quad \sum_u P_U(u) s^2(u) \leq R \quad \text{and} \\
			&\sum_u P_U(u) s^2(u) \leq \frac{\Theta}{\lambda_{\Lambda}},
		\end{align*}
		where the last inequality constraint follows from $\tau^{\top} \Lambda \tau = \lambda_{\Lambda}$.}\newline
	\textcolor{black}{Suppose that both inequality constraints are binding then it is evident that this can occur in the exceptional case for which $\lambda_{\Lambda} = \frac{\Theta}{R}$ and consequently the approximated secrecy capacity becomes:
		\begin{align}\label{binarySEC_Cap}
			\sum_u P_U(u) L^{\top}_u V L_u = \sum_u P_U(u) s^2(u)  \lambda_V = R \lambda_V,
		\end{align}
		which is achievable for $L_u = s(u) \tau$ with $\tau$ being orthogonal to $\sqrt{P_X}$ and the vector $s$ is orthogonal to $P_U(\cdot)$ with squared length equal to $R$.
		Next, consider the case when only the leakage constraint is active. Then $\rho = 0 $, $\nu = \frac{\lambda_{V}}{\lambda_{\Lambda}}$ and $\nu +\rho  - 1 \geq 0$ ~(or $\lambda_{V} \geq \lambda_{\Lambda}$). Then the inequality constraint require $R \geq \frac{\Theta}{\lambda_{\Lambda}} $, thus $\frac{\Theta}{R} \leq \lambda_{\Lambda} \leq \lambda_{V}$ and the approximated secrecy capacity becomes:
		\begin{align}\label{binarySEC_Cap_2}
			\sum_u P_U(u) L^{\top}_u V L_u = \sum_u P_U(u) s^2(u)  \lambda_V = \frac{\Theta}{\lambda_{\Lambda}} \lambda_V.
		\end{align}
		Finally, for the case in which only the rate constraint is active it follows that $\nu =0$, $\rho = \lambda_V $, with $\lambda_V \geq 1$ and $R \geq \frac{\Theta}{\lambda_{\Lambda}} $. Therefore $\lambda_{\Lambda} \leq \frac{\Theta}{R}$ and the approximated secrecy capacity becomes:
		\begin{align}\label{binarySEC_Cap_3}
			\sum_u P_U(u) L^{\top}_u V L_u = \sum_u P_U(u) s^2(u)  \lambda_V = R \lambda_V.
		\end{align}
	}
	Thus considering the above and Theorem~\ref{theorem_1} we have that
	\begin{equation}
		\label{sec_cap_app_bsc}
		C_s^{BSC} = \begin{cases}
			\frac{2}{\epsilon^2} R \lambda_{V} &\text{if $\lambda_{\Lambda} \leq \delta$,}\\
			\frac{2}{\epsilon^2} \frac{\Theta}{\lambda_{\Lambda}} \lambda_V &\text{if $\delta \leq \lambda_{\Lambda} \leq \lambda_{V},$}
		\end{cases}
	\end{equation}
	where $\delta = \Theta \mathbin{/} R$. Figure \ref{licsecvssec} compares the local approximated secrecy capacity with the secrecy capacity in~\cite{Csiz_Kor} for the BSWC, i.e. $C_s = H_b(q)- H_b(p)$, where $H_b(\cdot)$ is the binary entropy function.
	
	\begin{figure}[!htb]
		\includegraphics[width=\linewidth]{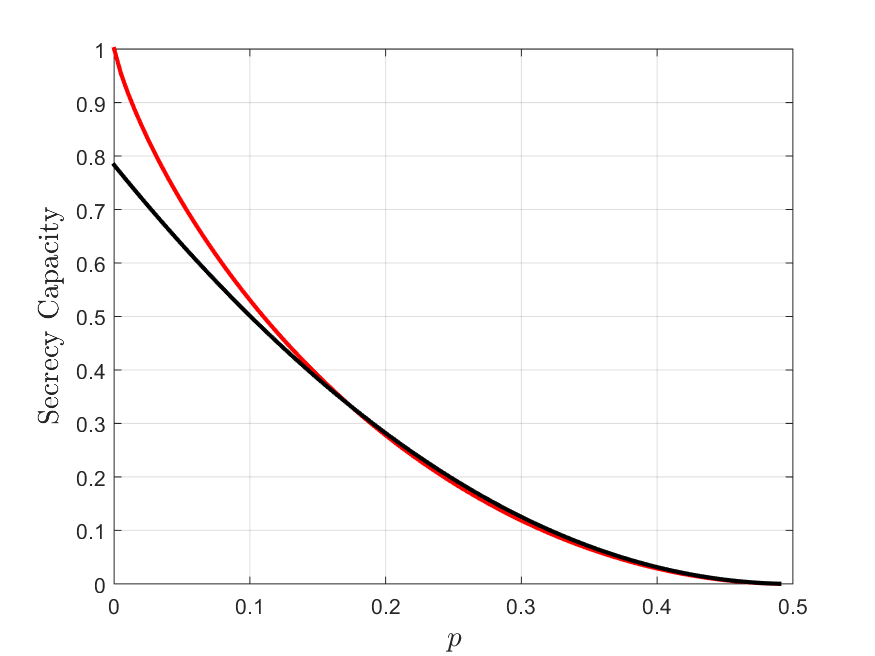}
		\caption{The black line is the approximation of \eqref{sec_cap_app_bsc} in comparison with $C_s = H_b(q)- H_b(p)$~(red line) for $\epsilon = 10^{-3}$, $\delta=0.085$ and $q=0.45$}
		\label{licsecvssec}
	\end{figure}
	
	\section{Discussion}
	In this work we have shown that local approximations based on Euclidean information theory can be utilized to linearize an information-theoretic security problem. Through a series of transformations we modified the standard wiretap channel into a simple linear algebra problem and derived an approximation of the secrecy capacity, which for the high secrecy regime seems to provide a good estimation of the regular secrecy capacity.

    \section*{Appendix A}
	\section*{Proof Lemma \ref{prop_1}}
     Let the distributions of interest $P_{X|U=u}$, $\forall u \in \mathcal{U}$ belong to the $\epsilon$-region of the reference pmf $P_X$, $\mathcal{N}_\epsilon^\mathcal{X}(P_X)$. Setting $Q_{X,u}^{(\epsilon)}=P_{X|U=u}=P_X+\epsilon J_u$ and expressing the mutual information for the encoding as 
		\begin{align}
  \label{proof_lem_1}
			I(U;X) &= \sum_{u\in \mathcal{U}} P_U(u) \cdot D(P_{X|U=u}|| P_X)\\
   \label{proof_lem_2}
            &= \sum_{u\in \mathcal{U}} P_U(u) \cdot D(Q_{X,u}^{(\epsilon)}|| P_X).
		\end{align}
		Utilizing equality~\eqref{KL_spherical_per} leads to the following
  \begin{align}
				I(U;X)=\frac{\epsilon^2}{2}\sum_{u\in \mathcal{U}} P_U(u) \cdot \parallel L_u \parallel^2+o(\epsilon^2), \label{prop_1_res1}
			\end{align}
			where $L_u = [\sqrt{P_X}]^{-1}\cdot J_u$ is the spherical perturbation vector in $\mathbb{R}^{|\mathcal{X}|}$. The mutual information between the input and the output of the channel $P_{Y|X}$ satisfies
			\begin{align}
				P_{Y|U=u} &= W^{(m)}P_{X|U=u} =W^{(m)}(P_X + \epsilon J_u) \nonumber \\ 
				&=P_Y + \epsilon W^{(m)} [\sqrt{P_X}] L_u, \label{utility_1}
			\end{align}
			where $W^{(m)} := P_{Y|X}$ is the $|\mathcal{Y}|\times |\mathcal{X}|$ transition matrix of the legitimate (main) channel and since $U\to X \to Y$ is a Markov chain. Using \eqref{utility_1} and \eqref{KL_spherical_per} we can rewrite $I(U;Y)$ as follows:
			\begin{subequations}
				\label{proof_utility:one}
				\begin{align}
					I(U;Y) &= \sum_{u\in \mathcal{U}} P_U(u) \cdot D(P_{Y|U=u}||P_Y) \nonumber \\
					&= \sum_{u\in \mathcal{U}} P_U(u) \sum_{y\in \mathcal{Y}} P_{Y|U=u} \log \frac{P_{Y|U=u}}{P_Y} \nonumber \\ 
					&= \sum_{u\in \mathcal{U}} P_U(u) \sum_{y\in \mathcal{Y}} P_{Y|U=u} \log \big(1+ \epsilon \frac{W^{(m)}J_u}{P_Y}\big) \nonumber \\ 
					&= \frac{1}{2}\epsilon^2 \sum_{u\in \mathcal{U}} P_U(u) \sum_{y\in \mathcal{Y}} \frac{(W^{(m)}J_u)^2}{P_Y} + o(\epsilon^2) \nonumber\\ 
					&=\frac{1}{2}\epsilon^2 \sum_{u\in \mathcal{U}} P_U(u) \cdot \parallel [\sqrt{P_Y}]^{-1}W^{(m)} [\sqrt{P_X}] \cdot L_u \parallel^2 \nonumber \\
                    &+ o(\epsilon^2) \nonumber\\
					& = \frac{1}{2}\epsilon^2 \sum_{u\in \mathcal{U}} P_U(u)\parallel B_{Y|X}L_u \parallel^2 + o(\epsilon^2) \Tsupersqueezeequ,
				\end{align}
			\end{subequations}
			where $B_{Y|X} = [\sqrt{P_Y}]^{-1}W^{(m)} [\sqrt{P_X}]$.\\
    \noindent Working with the Markov chain $U\to X \to Z$ along similar lines we obtain
    \begin{align}
				I(U;Z)= \frac{1}{2}\epsilon^2 \sum_{u\in \mathcal{U}} P_U(u)\parallel B_{Z|X}L_u \parallel^2 + o(\epsilon^2). \label{prop_3_res1}
			\end{align}
    Which completes the proof.

	\section*{Appendix B}
	\section*{Proof Theorem \ref{theorem_1}}
	
	\begin{proof}
		\noindent Differentiating the Lagrangian in \eqref{EqLagrangian_1} associated with problem \eqref{eq:optim5} with respect to $L_u$ gives:
			\begin{align}
				2P_U(u)K L^{\star}_u + \xi^{\star}(u)\sqrt{P_X} +  P_U(u) \mu^{\star} = 0 \quad \forall u \in \mathcal{U}. \label{pu1}
			\end{align}
			Then summing over all $u$:
			\begin{align}
				2 K \sum_u P_U(u) L^{\star}_u + \sum_u P_U(u) \xi^{\star}(u) \sqrt{P_X} + \mu^{\star} = 0. \label{pu2}
			\end{align}
			Recalling constraint \eqref{eq:const4_3}, the first term is zero and the above expression yields:
			\begin{align}
				\mu^{\star} = - \sum_u \xi^{\star}(u) \sqrt{P_X}. \nonumber
			\end{align}
			Substitution of \eqref{Eq_theorem_1} in \eqref{pu1} gives:
			\begin{align}
				K L^{\star}_u = c(u) \sqrt{P_X}, \label{pu3}
			\end{align}
			where \begin{align}
				c(u) = \bigg[ \frac{1}{2} \sum_u \xi(u) - \frac{\xi(u)}{2P_U(u)} \bigg]. \label{pu4}
			\end{align}
			We prove next that $c(u)=0$, for all $u \in \mathcal{U}$. \newline
	 
			Since $K$ is symmetric, it is orthogonally diagonalizable, i.e., it can be written as $K = Q \Delta Q^{\top}$ where $Q$ is an orthogonal matrix and $\Delta$ has the form:
			\begin{equation*}
				\def\arraystretch{1}
				\Delta =\begin{blockarray}{(ccc|ccc)l}
					\lambda_1&&0&\BAmulticolumn{3}{c}{\multirow{3}{*}{\huge$0$}}&\multirow{3}{*}{}\\&\rdots&&&&&\\0&&\lambda_k&&&&\\
					\cline{1-6}
					\BAmulticolumn{3}{c|}{\multirow{3}{*}{\huge$0$}}&0&&0&\multirow{3}{*}{}\\&&&&\rdots&&\\&&&0&&0&\\
				\end{blockarray}
				= \begin{blockarray}{(ccc|ccc)l}
					\Delta_1 & 0\\
					0 & 0
				\end{blockarray}
			\end{equation*}  
			where $\Delta_1>0$, since $\lambda_1>0$, for $i=\{1,2,...,k\}$ (because $K$ is nonnegative symmetric). Next we write \eqref{pu4} as:
			\begin{align}
				(Q \Delta Q^{\top}) L^{\star}_u &= c(u) \sqrt{P_X} \quad \text{or}  \label{new_KKT}\\
				\Delta \tilde{L}_u &= c(u) \tilde{p}, \label{new_KKT_1}
			\end{align}
			where $\tilde{L}_u = Q^{\top} L^{\star}_u$, $\tilde{p} = Q^{\top} \sqrt{P_X}$. 			
			We partition $Q$ conformably as  $Q^{\top}=\begin{bmatrix} Q^{\top}_1 \\Q^{\top}_2 \end{bmatrix}$ and similarly for $\tilde L_u=\begin{bmatrix}\tilde{L}^1_u \\ \tilde{L}^2_u \end{bmatrix}$ and $\tilde p=\begin{bmatrix}\tilde {p}_1 \\ \tilde {p}_2\end{bmatrix}$.  
			Then we have that
			\begin{align}
				\Delta_1 \tilde{L}^1_u = c(u) \tilde {p}_1 \quad \text{and} \quad 0=c(u) \tilde {p}_2. \label{KKT_new_2}
			\end{align}
			Suppose that $c(u)\neq 0$ then $\tilde{p}_2 = 0$ and considering the orthogonality constraint \eqref{eq:const5_2} gives:
			\begin{align}
				\sqrt{P_X}^{\top} L_u = \sqrt{P_X}^{\top} Q Q^{\top} L_u = \tilde{p}^{\top} \tilde{L}_u = \tilde{p}^{\top}_1 \tilde{L}^1_u.
			\end{align}
			But $\tilde{L}^1_u = - \Delta^{-1}_1  c(u) \tilde {p}_1 $ hence $(\tilde {p}^{\top}_1 \Delta^{-1}_1 \tilde {p}_1)c(u) = 0$, implying that $\tilde {p}_1 = 0$ and then $\tilde{p} = Q^{\top} \sqrt{P_X} = 0$ or $ \sqrt{P_X} = 0$, which is a contradiction. Therefore, $c(u)= 0$, for all $u \in \mathcal{U}$ and the proof is completed.
			
\end{proof}
	
\end{document}